\documentclass[12pt]{article}
\usepackage[top=1.25in, bottom=1.25in, left=1.25in, right=1.25in]{geometry}
\usepackage{amsmath}
\usepackage{amsthm}
\usepackage{amssymb}
\usepackage{graphicx} 
\usepackage{moresize}
\usepackage{bbm}
\usepackage{parskip}% http://ctan.org/pkg/parskip 
\usepackage{adjustbox}
\usepackage{lipsum}
\usepackage{geometry}
\usepackage[utf8]{inputenc}
\usepackage{url}
\usepackage{float}

\makeatother
\usepackage{setspace}
\usepackage{thmtools}
\newtheorem{proposition}{Proposition}

\newtheorem*{theorem*}{Theorem 1}
\newtheorem*{corollary*}{Corollary 1}
\newtheorem*{corollary**}{Corollary 2}

\newtheorem*{assumption*}{Assumption 1}
\newtheorem*{assumption**}{Assumption 2}

\title{ADESS: A Proof-of-Work Protocol to Deter Double-Spend Attacks}

\author{Daniel Aronoff\thanks{Massachusetts Institute of Technology} \and Isaac Ardis\thanks{Ethereum Classic Cooperative}}

\begin{document}

%\authorrunning{D. Aronoff and I. Ardis} 
 
%\mainmatter

%\author{Daniel Aronoff\and Isaac Ardis}

% \institute{Massachusetts Institute of Technology, Cmabridge MA, USA
% \email{daronoff@mit.edu},
% https://dci.mit.edu/daniel-aronoff\and
% Ethereum Classic Cooperative}

% \authorrunning{D. Aronoff and I. Ardis} 

\maketitle

{
  \renewcommand{\thefootnote}{}
  \footnotetext{\vskip 0.5em
    \noindent Corresponding author Daniel Aronoff, daronoff@mit.edu. \\
    \noindent Copyright \copyright{Daniel Aronoff and Isaac Ardis}}
}

\pagenumbering{roman}

\begin{abstract}
A principal vulnerability of a proof-of-work ("PoW") blockchain is that an attacker can re-write the history of transactions by forking a previously published block and build a new chain segment containing a different sequence of transactions. If the attacker’s chain has the most cumulative mining puzzle difficulty, nodes will recognize it as canonical. We propose a modification to PoW protocols, called $ADESS$, that contains two novel features. The first modification enables a node to identify the attacker chain by comparing the temporal sequence of blocks on competing chains. The second modification penalizes the attacker by requiring it to apply exponentially increasing hashrate in order to make its chain canonical. We demonstrate two things; (i) the expected cost of carrying out a double-spend attack is weakly higher under $ADESS$ compared to the current PoW protocols and (ii)  for any value of transaction, there is a penalty setting in $ADESS$ that renders the expected profit of a double-spend attack negative.
\end{abstract}
\bigskip
\bigskip

\vskip5pt

%\ccsdesc[500]{Networks~Network protocol design}
\thispagestyle{empty}
\pagebreak
\newpage
\onehalfspacing
\pagenumbering{arabic}

\section{Introduction}
The protocol currently used in Bitcoin, Ethereum Classic and most other PoW cryptocurrencies instructs nodes to recognize as canonical the chain with the highest score, which is defined as cumulative mining puzzle difficulty (also referred to as "work"). A principal vulnerability of a blockchain is that an attacker can re-write the history of transactions by forking a previously published block and building a new chain segment containing a different sequence of transactions. If the attacker chain (denoted $\mathcal{A}$) has more cumulative mining puzzle difficulty compared to the incumbent canonical chain (denoted $\mathcal{IC}$), the PoW protocol instructs nodes to recognize the attacker chain as canonical. One motivation for forking a blockchain is to carry out a double-spend attack, whereby the attacker negates a transfer of its tokens that are recorded on chain $\mathcal{IC}$. The safety threat posed by double-spending is not just hypothetical. There have been instances of sizable attacks that have succeeded. For example, from 2018 to 2020 there were several double spend attacks in Ethereum Classic and Bitcoin Gold\footnote{For Ethereum Classic see Andrew Singer \cite{Singer (2020)} and James Lovejoy \cite{Lovejoy (2020)}. For data on double-spend attacks on PoW cryptocurrencies prior to 2020, see the MIT Digital Currency Initiative 51\% reorg tracker (2020) \cite{MIT DCI}}. The double-spend vector of attack raises questions about the security of transactions on a PoW blockchain.

We propose a modification to the protocols introduced by Satoshi Nakamoto in the Bitcoin White Paper (Nakamoto (2008) \cite{Bitcoin White Paper (2008)} and by Gavin Wood in the Ethereum "Yellow Paper" (Wood (2021) \cite{Eth Yellow Paper (2021)}), which we refer to as the "Nakamoto" protocol. We name our protocol modification Absolute Discontinuous Exponential Subjective Scoring ("$ADESS$"), which is designed to increase security against double-spend attacks by raising the cost to make a fork chain canonical without significantly compromising other key dimensions of security and performance of the network and without using any new oracles. $ADESS$ introduces two key modifications to current PoW protocols. One is a criteria for identifying an attacker chain which, unlike current $PoW$ protocols, requires nodes to observe the temporal sequence of blocks. The other is an altered criteria for scoring chains after the attacker has been identified. $ADESS$ is most effective in raising the cost of an attack when the mining puzzle difficulty is adjusted between short intervals of blocks.

\subsection{The two \textit{ADESS} modifications}

The first $ADESS$ modification is a criteria for identification of an attacker chain based on behavior associated with a double-spend attack. The intended victim, Bob, must believe that the transaction sending him tokens is appended to the canonical chain before he conveys an exchange item to the sender, Alice. When he observes the transaction appended to a block, Bob will wait until a few more blocks have been appended to the chain to confirm that the chain remains canonical before conveying the item to Alice. At the same time Alice does not want Bob to know that she is building her own chain, so she will not broadcast her chain until after she has received the item from Bob. We use Alice's delay in broadcasting her chain to form a criteria for identification. Roughly, when comparing two chains with a common ancestor block (the "fork-block"), a penalty is assigned to the chain that was last to broadcast a minimum number of successive blocks, starting from the fork-block.

% Identification under $ADESS$ requires that nodes track the temporal sequence of observed blocks. This is another departure from the Nakamoto protocol, where a node need only compare the cumulative mining puzzle difficulty to decide which chain is canonical.

The second $ADESS$ modification is to penalize the identified attacker chain. When chain $\mathcal{A}$ is identified the criteria for choosing the canonical chain is changed.  Chain $\mathcal{IC}$ and chain $\mathcal{A}$ are assigned scores based on the number of post-fork blocks, with a discount applied to each block in chain $\mathcal{A}$. This requires chain $\mathcal{A}$ to grow at a faster rate than chain $\mathcal{IC}$ to become canonical. To speed up the growth rate, an attacker must increase the amount of hashrate, i.e. puzzle solution guesses per unit of time, on chain $\mathcal{A}$ in expectation. The elevated growth rate causes an increase in mining puzzle difficulty, which requires the attacker to increase hashrate further after each puzzle adjustment in order to maintain the elevated chain growth rate. This process leads to an exponential increase in the hashrate an attacker must apply to exceed its penalty growth rate. Hashrate consumes electricity, so the attacker's cost increases exponentially over time.  The extent to which mining expenditure exceeds the block reward is a sunk cost which the attacker cannot recover even if its chain becomes canonical and it receives block rewards. 

The purpose of $ADESS$ is to confront a double-spend attacker with the prospect, in terms of ex-ante expected value, of an exponentially increasing non-recoverable cost to carry out its attack.

\subsection{Roadmap}

The rest of the paper is organized as follows. In Section \ref{sec: Related Work - The Double -Spend Vulnerability} we review the Nakamoto protocol for validating transactions in PoW blockchains and we describe a prototypical double-spend attack. Then we characterize the claim made by  and Budish \cite{Budish (2018)} and Gervais et.al. \cite{Gervais et.al.} that a double-spend attack can be low cost, followed by the rebuttal of Moroz et.al.'s \cite{Moroz et.al. (2020)} that retaliation by an intended victim induces a war of attrition between attacker and victim in which the outcome is uncertain. We conclude that Nakamoto neither promotes, nor discourages, double-spend attacks. Section \ref{sec: The $ADESS$ Protocol} begins with a statement of the intended goal of $ADESS$ and then describes the two parts of the $ADESS$ protocol modification. The first part is the rule for assigning a penalty to a fork chain. The second part is the operation of the penalty. Section \ref{subsec: The intrinsic deterrent to double-spend attacks} states Proposition 1, which demonstrates that the cost of carrying out a double-spend attack under $ADESS$ is weakly higher compared to Nakamoto. In Section \ref{sec: The Model} we present a baseline model for evaluating the properties of $ADESS$ and characterize the attacker's decision problem. Section \ref{sec: Some Properties of the $ADESS$ Protocol} states Theorem 1, which demonstrates that the $ADESS$ penalty can be tuned to render a double-spend attack ex-ante unprofitable for any value of transaction. Section \ref{sec: Relaxing the Baseline Model Constraints} examines the implications, in terms of blockchain performance, of relaxing two of the four constraints that were placed on the baseline model; partial adjustment of mining puzzle difficulty and network latency. In Section \ref{sec: Additional Considerations} we compare the performance of $ADESS$ and Nakamoto along two security related dimensions; malicious attacks and unobserved forks. In the Appendix, we examine the implications of relaxing the remaining constraints that were placed on the baseline model; multiple chains and non-constant blockchain growth rates. 

\section{Related Work: The Double-Spend Vulnerability}
\label{sec: Related Work - The Double -Spend Vulnerability}

In this section we first describe the mechanics of a double-spend attack. This serves as as a foundation for a review of contributions to the literature on the vulnerability of a $PoW$ blockchains to double-spend attacks, which provide the background and motivation for our proposal to modify the Nakamoto protocol.

\subsection{A double-spend attack under the Nakamoto protocol}
\label{subsec: A Double - Spend Attack}

We employ the following framework. The first block of a blockchain is called the "genesis block" and is assigned the number 0. The number assigned to the child block is the parent block number plus 1. A blockchain has a tree structure in which the genesis block is the root and chain segments form branches that can fan out as the chain grows. A fork-block is the common ancestor block of two or more post- fork chain segments (which we refer to as "chains"). Block \#98  in Figure \ref{fig:A Double-Spend Attack} is the common ancestor of chains $\mathcal{IC}$ and $\mathcal{A}$. Two blocks on different post-fork chains are assigned the same number if each block is the same number of blocks away from the Genesis block. In Figure \ref{fig:A Double-Spend Attack} each chain has a block \# 100. We assume no latency. Each node receives broadcasts instantaneously. 

A double-spend attack is an operation whereby a node can negate the transfer of tokens it sent to a counter-party after receiving the exchange item from the counter-party. Here is how an Alice node double-spends a Bob node. Alice sends a token to a counterparty, Bob, which is appended to block \#100 on chain $\mathcal{IC}$. Alice secretly builds a chain forking from an earlier block, \#98, applies more hashrate to her chain $\mathcal{A}$ compared to chain $\mathcal{IC}$, in order to ensure that her chain grows faster in expectation, and appends to a block in chain $\mathcal{A}$ a transaction whereby she sends all of the tokens in the wallet from which she paid Bob, to other wallets, possibly owned by her. After Alice receives the exchange item from Bob - which may take place only after Bob observes several child blocks - provided chain $\mathcal{A}$ has more blocks, she broadcasts chain $\mathcal{A}$. Thereupon, in accordance with the Nakamoto protocol, the other nodes accept Alice's chain as the true, or 'canonical' chain. Because Alice emptied her wallet on an earlier block, her transfer to Bob is no longer valid.\footnote{Alice's transfer to Bob remains in the pool of transactions that could be appended by a miner to a block in chain $\mathcal{A}$. However, since her original wallet is empty, the transaction would be invalid, as there are no tokens in her wallet to send to Bob.} Alice has succeeded in retrieving her tokens after collecting the exchange item from Bob. The double-spend attack is displayed in Figure \ref{fig:A Double-Spend Attack}. Note that, in order for the attacker chain $\mathcal{A}$ to become canonical, it must have more cumulative mining puzzle difficulty at some future time, starting from block \# 98, compared to the incumbent chain $\mathcal{IC}$. \footnote{The ranking of chains in terms of mining puzzle difficulty does not necessarily match the ranking in terms of hashrate. The reason is that the rate of solving puzzles is stochastic. A lucky string of quickly solved puzzles on one chain will generate more blocks - and more cumulative puzzle difficulty - on that chain compared to another chain with more hashrate but less good luck in solving puzzles.}

\begin{figure}
\begin{center}
\includegraphics[page=1,width=0.9\textwidth,height = 0.22
\textheight]{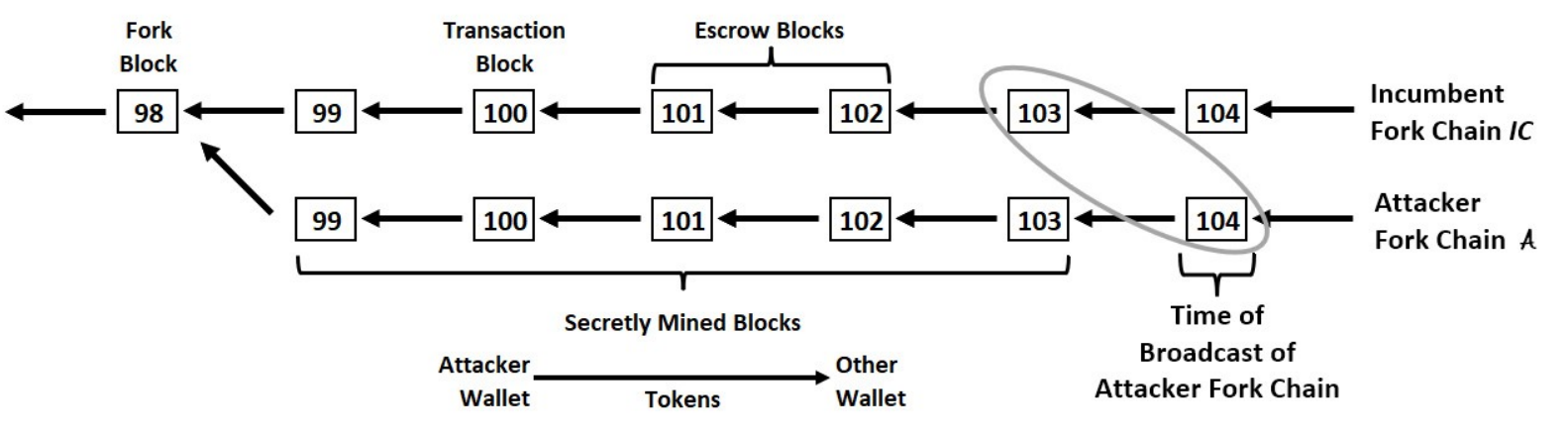}
\end{center}
\caption{A Double-Spend Attack}
\label{fig:A Double-Spend Attack}
\end{figure}

\subsection{Evaluating Nakamoto's claim of double-spend deterrence}
\label{subsec: Nakamoto's Claim and Budish's Critique}

In the 2008 Bitcoin White Paper (Nakamoto (2008)) \cite{Bitcoin White Paper (2008)}, Satoshi Nakamoto recognized the vulnerability of the protocol to the possibility of a double-spend attack. Nakamoto offered two arguments in support of the idea that the vulnerability is limited. The first argument is that the cost of to acquiring the hashrate required to carry out a double-spend attack is a deterrent.\footnote{Nakamoto (2008) \cite{Bitcoin White Paper (2008)} section 11, pp. 6-7.} The second argument is that a miner with sufficient hashrate to carry out a double-spend attack would not risk the depreciation in the exchange value of the currency that  a double-spend attack might cause\footnote{Nakamoto (2008) \cite{Bitcoin White Paper (2008)} section 6, p. 4 states "[A miner] ought to find it more profitable to play by the rules, such rules favor him with more new coins than everyone else combined, than to undermine the system and the validity of his own wealth."}. Budish (2018) \cite{Budish (2018)} and Gervais et.al. (2016) \cite{Gervais et.al.} cast doubt on the former claim by showing that a double-spend attack on Bitcoin could be profitable for a relatively modest value  transaction.\footnote{Budish (2018) \cite{Budish (2018)}) also pointed out that a short seller of Bitcoin could profit from a disruption that caused the exchange value to decline.} 

We construct an example to elucidate this point. To become canonical chain $\mathcal{A}$'s work must exceed chain $\mathcal{IC}$'s work at some time. When chain $\mathcal{A}$ becomes canonical, the attacker retains the double-spend transaction $v$ and it will receive the block rewards on chain $\mathcal{A}$. Equation \ref{eq: The Budish Model} is the ex-ante expected profit of a double-spend attack where chain $\mathcal{A}$ has $N$ post-fork blocks at the time it becomes canonical. It reflects an attacker strategy of forking the head (i.e. rightmost block) of the canonical chain when it sends the transaction into the transaction pool and applying the same hashrate to chain $\mathcal{A}$ as is applied to chain $\mathcal{IC}$ (normalized to 1 unit of hashrate per block) until the victim conveys the exchange item - at block $N -1$. In expectation this implies that chain $\mathcal{A}$ will grow at the same rate and accumulate the same mining puzzle difficulty through post-fork block $N-1$ (which is normalized to 1 block per unit of time). after block $N-1$ the attacker adds an additional $\epsilon$ of hashrate to ensure that chain $\mathcal{A}$ has greater cumulative mining puzzle difficulty in expectation when it is broadcast at block $N$. We study this strategy because it can be implemented on any PoW blockchain regardless of the frequency or rate of adjustment of mining puzzle difficulty.

The notation is as follows. $p_{B}$ - assumed constant for all blocks -  is the block reward paid to the miner who is first to solve the puzzle. $c$ is the attacker's cost, in dollars, of a unit of hashrate - defined as the cost of operating one mining computer for 1 unit of time. $\delta \in (0,1]$ is the discount rate per unit of time. The cost $c$ is incurred at each block and the revenue is realized at the end, when chain $\mathcal{A}$ is broadcast. The crypto to dollar exchange rate is 1. There is no latency in the network.

\begin{multline}
\label{eq: The Budish Model}
\text{Attacker ex ante expected profit} =\\
\text{E}[\underset{\text{discounted revenue}}{\underbrace{\delta^{N-1}(v + p_{B}N)}} - \underset{\text{discounted cost}}{\underbrace{c\big[\sum_{n = 1}^{N-1}\delta^{n - 1} + (1 + \epsilon)\delta^{N-1}\big]}}]
\end{multline}

The minimum profitable transaction can be expressed by the following inequality.

\begin{center}
$v > \delta^{-(N - 1)}c\Big\{\big[\sum_{n = 1}^{N-1}\delta^{n - 1} + (1 + \epsilon)\delta^{N - 1}\big] - \delta^{N -1}Np_{B}/c\Big\}$
\end{center}

Since the time interval between blocks is short - approximately 13 seconds for Ethereum and 10 minutes for Bitcoin - the time discount factor will be very close to 1. Taking the limit as $\delta \rightarrow 1$ yields the expression for minimum profitable transaction size.

\begin{equation}
\label{eq: minimum profitable transaction size - general case}
v > (c - p_{B})N + c\xi
\end{equation}

The first term reflects the attacker's profit from mining chain $\mathcal{A}$. The second term is the cost of the additional hashrate applied to the $N$th post-fork block. When the attacker and incumbent hashrate cost is identical, $(c = p_{B})$ is the market rate of profit and can be normalized to zero. In that case, the expression for minimum profitable transaction size is

\begin{equation}
\label{eq: minimum profitable transaction size - identical costs}
v > c\xi
\end{equation}

In Equations \ref{eq: minimum profitable transaction size - general case} and \ref{eq: minimum profitable transaction size - identical costs} $\xi$ can be arbitrarily small. We conclude that the minimum transaction size required to render the ex-ante expected value of double-spend attack profitable is an increasing function of the spread between the cost of hashrate and the block reward $(c - p_{B})$.  When the attacker and incumbent hashrate cost is identical, the transaction size can be arbitrarily small.\footnote{Another dimension from which the double-spend attack can be analyzed is the expected profit from an attack using less than the honest miner hashrate. For example, an attack in which 30\% of the hashrate applied to chain $\mathcal{IC}$ was applied to chain $\mathcal{A}$ would, in expectation, recoup 30\% of its hashrate cost from block rewards and earn a third of the value of the transaction. In our example, the minimum value of the transaction required to make the attack profitable would be $ v > 3\big(c\xi + (2/3)N\big)$. See Gervais et.al. (2016) \cite{Gervais et.al.} for a detailed simulations in a Markov Decision process framework. Gervais et.al. also evaluate the relationship of the stale block rate and the block reward to the expected profit of a double-spend attack. We do not evaluate those dimensions.}

\subsection{Retaliation}
\label{subsec: Retaliation}

Moroz et.al (2020) \cite{Moroz et.al. (2020)}. pointed out that a victim of a double-spend attack could retaliate by increasing the hashrate applied to the incumbent chain $\mathcal{IC}$. The possibility of retaliation means the blockchain may be less vulnerable to attack than is implied by the analyses of Budish or Gervais et.al., neither of which address the possibility of retaliation. In the resulting war of attrition both attacker and victim are confronted with a similar marginal decision of whether to add another block\footnote{In the Moroz et.al (2020) \cite{Moroz et.al. (2020)}. model, it is assumed that the victim and the attacker build the chains. The other miners step aside until the bifurcation is resolved.}. In the model each numbered block (one for each chain) constitutes a round and the two chains add blocks at the same constant rate.\footnote{The Moroz et.al. (2020) model assumes the puzzle difficulty remains constant, which implies $\gamma = 0$.} At the end of a round, each player decides whether to mine a block and enter the next round. Equation \ref{eq: The Moroz Model} depicts a variant of the {Moroz et.al. (2020)} model. The expected payoff, for either player entering the next round $(N + 1)$, assuming its opponent exits, is

\begin{equation}
\label{eq: The Moroz Model}
\text{E}[\underset{\text{revenue from becoming canonical}}{\underbrace{v + p_{B}(N + 1)}} - \underset{\text{cost of mining next block}}{\underbrace{c(1+\gamma)^{N+1}}}]
\end{equation}

There are equilibrium strategies that support a range of potential outcomes in this symmetric war of attrition. To achieve a determinate outcome would require a restriction on the strategy space or an equilibrium refinement. 

\subsection{A deterrence deficit}
\label{subsec: A deterrence deficit}

Nakamoto's claim that Bitcoin (and by extension, PoW blockchain protocols generally) provide incentives that minimize the security risk of a double-spend attack is not supported. Under free-entry and uniform hashrate costs a profit could be made from double-spending a transaction of any size. Moreover, when an attack is carried out, the possibility of counter-attack by the victim renders the outcome uncertain. The conclusion is that Nakamoto on its own neither deters nor promotes double-spend attacks. Any deterrence must come from outside the protocol. For example, a limitation on mining computers that elevates the attacker's hashrate cost above the honest miner's cost; or financing constraints that limit the capital an attacker can use to carry out a attack. We next begin our discussion of $ADESS$, which incorporates a double-spend attack deterrent inside the protocol.

\section{The $ADESS$ Protocol Modification}
\label{sec: The $ADESS$ Protocol}

$ADESS$ modifies Nakamoto by adding two protocols. The first addition is a  criteria for a node to identify the attack chain and assign a penalty. This is the first subjective discontinuity of $ADESS$. The second addition is the application of the penalty  to the attack chain. The attack chain's score changes from cumulative mining puzzle difficulty to weighted chain length, where the length of the penalized chain is discounted. When a node observes the discounted length of the penalized chain to exceed the un-discounted length of the incumbent canonical chain, the attack chain becomes canonical and its score switches back to cumulative mining puzzle difficulty. This is the absolute exponential scoring and the second subjective discontinuity of $ADESS$. 

\subsection{The penalty assignment}
\label{The Penalty Assignment}

The objective is to assign the penalty to the chain that is built by the double-spend attacker. The proposed assignment rule is  designed to optimize over two goals.\vskip5pt  

(i) Minimize the likelihood that the penalty is assigned to the chain to which the transaction is appended,and 

(ii) Maximize the likelihood that the penalty is assigned to the attacker's chain. 

\subsubsection{The victim's incentive to wait for transaction confirmation}

 When the transfer of tokens is first appended to a block on the canonical chain, the recipient of the tokens has an incentive to require that the chain remain canonical for one or more additional confirmation blocks before the exchange item is released  to its counterparty ( we refer to the number of confirmation blocks, inclusive of the block to which the transaction is appended, as "confirmation depth"). One motivation is to ensure that the receipt of tokens by a node is not later negated as a result of another chain becomes canonical. So-called "uncle-chains" are forks off of the canonical chain that are typically abandoned after a few blocks as consensus forms on a single chain. Uncle-chains arise with considerable frequency on some networks. For example, approximately 8.5\% of valid blocks produced in the Ethereum network since inception are part of post-fork chains that were eventually abandoned.\footnote{\url{https://etherscan.io/uncles}} \footnote{Another reason for increasing confirmation depth may be to protect against the possibility of a double-spend attack. However,as we demonstrated in Section \ref{subsec: Nakamoto's Claim and Budish's Critique}, Nakamoto does not intrinsically deter double-spend attacks at any confirmation depth. Any security provided by increasing confirmation depth under Nakamto arises from the contingent circumstances in which the transaction takes place.} On the other hand, the recipient will prefer to complete the transaction sooner rather than later. 
 
 This suggests a tradeoff between the recipient's desire for security - which increases with the passage of time -  and the discounted value of the transaction - which decreases with the passage of time.  Guo and Ren (2022) \cite{Guo and Ren} elegantly formalized and proved the tradeoff as a relation between confirmation blocks $k$, block propagation delay upper bound $\Delta$, mining rate $\lambda$ and the fraction of honest hashrate $\rho$. Equation \ref{eq: Security Tradeoff} is an upper bound on the probability that the recipient's receipt of tokens is later negated when a majority of hashrate is controlled by honest miners (subject to restrictions on $\lambda\Delta$). \footnote{Equation \ref{eq: Security Tradeoff} is part of Theorem 1 of Gou and Ren (2022). The precondition in Theorem 1 requires that $\lambda\Delta < ln(\rho)ln(1/2)$.}

 \begin{equation}
 \label{eq: Security Tradeoff}
\Big (2 + 2\sqrt{\frac{1}{p - 1}}\Big )4p(1-p)^{k},\;\; \text{where}\; p = \rho\epsilon^{\lambda\Delta}
 \end{equation}

{\noindent{When the attacker applies less than half the hashrate $\rho$, an increase in the number of confirmation blocks $k$, improves security by decreasing exponentially the probability that the receipt of tokens by a node is later reversed (i.e. by decreasing Equation \ref{eq: Security Tradeoff}). We define the lower bound to confirmation depth $k = \alpha$ for a node as that which optimizes the tradeoff between security and timeliness when the node places a zero probability of being the target of a double-spend attack, i.e. $\rho = 1$. This forms a lower bound to the confirmation depth the node actually chooses. When a node is concerned that it might be the target of an attack, $\rho < 1$, it will need to increase its confirmation depth to maintain its security level. For the purposes of our analysis we will assume that $\alpha$  is the same for all nodes.\footnote{Alternatively, $\alpha$ can be viewed as the lower bound of node confirmation blocks.}

 \subsubsection{The attacker's incentive to build its chain in secret}
 \label{subsec: The attacker's incentive to build its chain in secret}

The attacker has an incentive to privately build its chain until after it has received the exchange item from its victim for two reasons. First, the broadcast of the attacker chain reveals that the attacker has emptied its wallet on the chain, which could alert the victim to the attack. Second, if the attacker broadcast a chain that had the most work, it would be canonical and the block to which the victim's receipt of tokens is appended would become part of an uncle chain, which would nullify the victim's receipt of tokens. Recall from Figure \ref{fig:A Double-Spend Attack} that the fork block for the attacker's chain $\mathcal{A}$ must be at a lower height than the block containing the transaction. Therefore, at the minimum, an attacker will not broadcast its chain for an interval of $\alpha$ blocks, counting from the fork-block.\footnote{Guo and Ren \cite{Guo and Ren} prove that any successful double-spend attack strategy can be carried out privately. It follows that an attack strategy that is optimal when the attacker chain is built in secret, is globally optimal.} Given the assumption that $\alpha$ is a lower bound to the confirmation blocks of all nodes, it follows that no double-spend attacker will broadcast its chain until after its victim has observed $\alpha$ confirmation blocks.

\subsubsection{The Penalty Assignment Rule}
These considerations suggest that (i) and (ii) can be optimized by exempting from the penalty the first chain that is observed to have an interval of $\alpha$ valid post-fork blocks and penalizing all other chains that start from the same fork block. This assignment rule reflects the equilibrium behavior of the attacker and the victim. At $\alpha$ blocks, the frequency of uncle blocks is low, which minimizes the probability that the transaction is on a chain that will subsequently be abandoned (goal (i)). The attacker will not broadcast its chain before $\alpha$ confirmation blocks have appeared, which ensures it will be penalized (goal (ii)). 

\textbf{Penalty Assignment Rule}
When a node observes two chains with a common fork-block ancestor and at least one of the chains is of length $\alpha$ blocks, the chain that was first observed to achieve length $\alpha$ blocks is called "chain $\mathcal{IC}$". The other chain is called "chain $\mathcal{A}$". The penalty is assigned to chain $\mathcal{A}$. From the time of the penalty assignment to the time of the removal of the penalty, chain $\mathcal{IC}$ is canonical.\qed

\subsection{The penalty function}

When a penalty assignment has been made, $ADESS$ diverges from the Nakamoto protocol. The canonical chain is chosen by comparing the number of post-fork blocks between chain $\mathcal{A}$ and chain $\mathcal{IC}$, rather than cumulative mining puzzle difficulty. In order to become canonical, the length of a chain $\mathcal{A}$ must exceed  chain $\mathcal{IC}$ by some number of blocks. We refer to the manifold formed by the combinations of blocks - on chain $\mathcal{A}$ and chain $\mathcal{IC}$ - at which chain $\mathcal{A}$ becomes canonical, as the "canonical boundary". When a penalized chain $\mathcal{A}$ has crossed the canonical boundary, the penalty ceases to apply. The protocol for determining the canonical chain reverts to a comparison of cumulative mining puzzle difficulty.\footnote{The $ADESS$ criteria for comparing chains after chain $\mathcal{A}$ has reached the canonical boundary is explored in Appendix \ref{app: Generalizing $ADESS$ to multiple chains and multiple forks}.}

\textbf{Penalty Function}
The scoring formula has two elements.

(i) The scoring criteria for chain $\mathcal{IC}$ and chain $\mathcal{A}$ switches from cumulative mining puzzle difficulty to number of blocks. A per-block penalty weight of $\frac{1}{1 + \xi}$ is applied to each block on chain $\mathcal{A}$, starting at its first post-fork block and continuing until the time that chain $\mathcal{A}$ is observed to have reached the canonical boundary. During the time interval that the penalty is applied to chain $\mathcal{A}$, the  blocks on chain $\mathcal{IC}$ are unweighted. The score of an interval of $N$ blocks on a chain $\mathcal{A}$ is $\frac{1}{1 + \xi}N$ compared to a score of $N$ for chain $\mathcal{IC}$.\vskip5pt 

(ii) After chain $\mathcal{A}$ has reached the canonical boundary, the protocol reverts to the Nakamoto criteria of comparing cumulative mining puzzle difficulty between the two chains, with a re-set of chain $\mathcal{A}$'s score. At the point of crossing, the score of chain $\mathcal{IC}$ is re-set to equal the cumulative mining puzzle difficulty of chain $\mathcal{IC}$, plus a small additional amount $\epsilon$ - which makes chain $\mathcal{A}$ canonical.\qed

Figure \ref{fig:The Penalty Function} displays the two elements of the penalty function. No penalty is assigned until $\alpha$ post-fork blocks on chain $\mathcal{IC}$ have been appended. When that occurs, the penalty is assigned to chain $\mathcal{A}$, starting retroactively from the fork-block. In order to become canonical, chain $\mathcal{A}$ must have $(1 + \xi)N$ blocks at a time when chain $\mathcal{IC}$ is of length $N$.\footnote{It is worth pointing out that $ADESS$ conforms to the Axioms of Leshno and Strack (2020) \cite{Leshno and Strack} since it does not alter the underlying Nakamoto entry and reward structure for miners.}

\begin{figure}[H]
\begin{center}
\includegraphics[page=1,width=0.5\textwidth,height = 0.3\textheight]{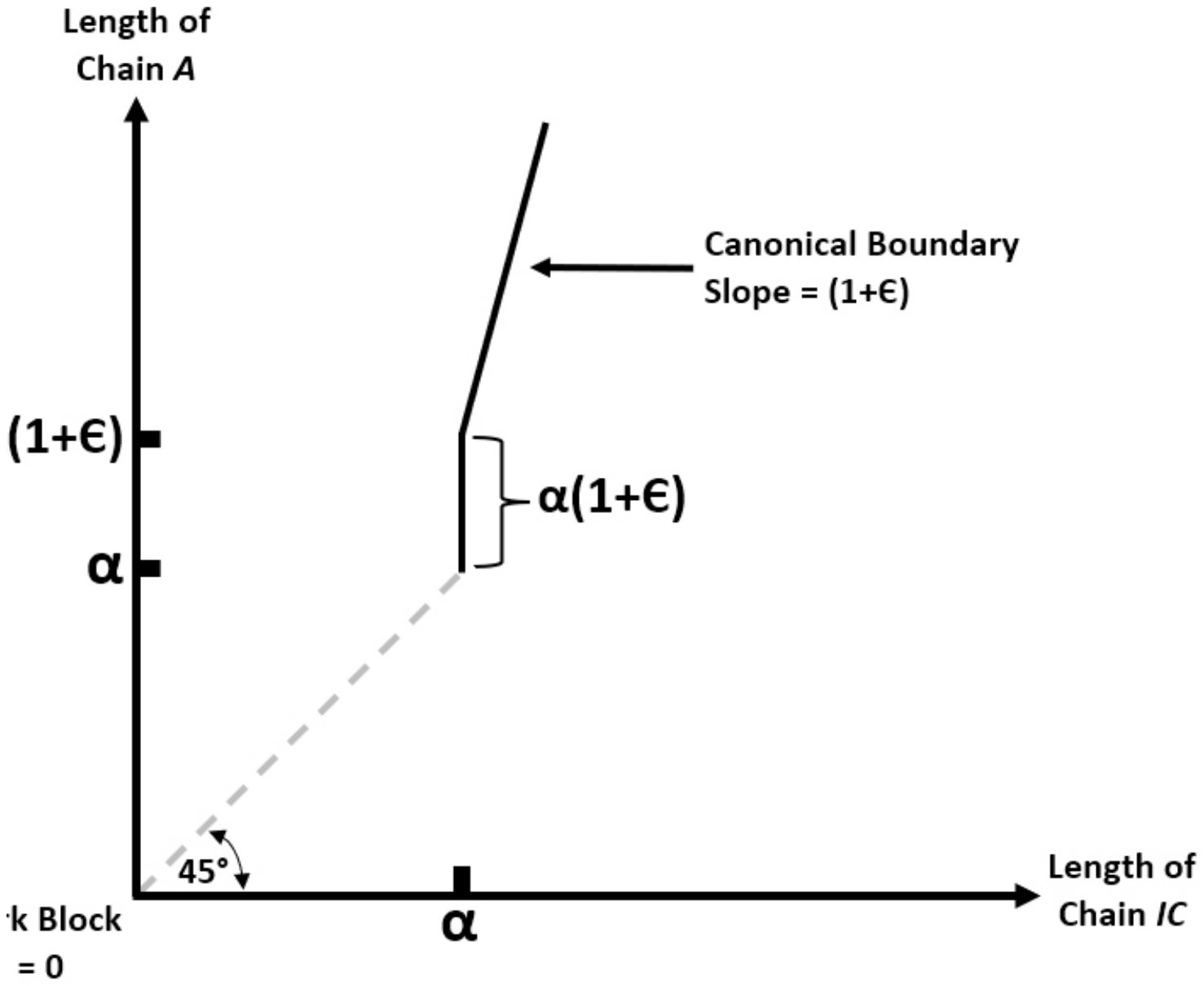}
\end{center}
\caption{The Penalty Function}
\label{fig:The Penalty Function}
\end{figure}

\subsection{The exponentially increasing sunk cost of a double-spend attack} 

Consider the case where there is no latency, target time interval and initial hashrate are 1 and puzzle difficulty fully adjusts after each block. In that case, the expected hashrate required to append the next block in the target time interval is equal to the expected hashrate applied to the prior block. Suppose the attacker applies a constant growth rate of $\gamma$.\footnote{The target blockchain growth rate $T$ can be represented as a Bernoulli Process $\mathbb{E}[T] = D/h$, where $D$ denotes mining puzzle difficulty; the probability of guessing the puzzle solution is $1/D$ and hashrate is $h$. In our example $T =1$ and initial $h = 1$. The Bernoulli Process yields $D =1$ at the first block. Full adjustment implies that an increase of hashrate to to $1+ \gamma$ causes difficulty to increase to $D = 1 + \gamma$ and so forth.} The interaction of a growth rate above target with the mining puzzle difficulty adjustment imposes an ex-ante expected exponentially increasing cost to carry out a double-spend attack. At the first post-fork block the attacker applies $(1+\gamma)$ units of hashrate. At the second block, after the mining puzzle difficulty  has increased to require $(1+\gamma)$ units of hashrate to achieve the target, the attacker must apply $(1 + \gamma)^{2}$ units of hashrate to maintain its growth rate. At each post-fork block $n$ the attacker must apply $(1 + \gamma)^{n}$ units of hashrate. The increase in the attacker's per block cost above the block reward $p_{B}$ is a sunk cost that it cannot recover from block rewards  when chain $\mathcal{A}$ crosses the canonical boundary. Figure \ref{fig: The exponential sunk cost of a double-spend attack} displays the sunk cost that the exponentially increasing penalty imposes on the attacker. It is the ex-ante expectation of exponentially increasing sunk cost that provides the deterrent to launching a double-spend attack.

\begin{figure}[H]
\begin{center}
\includegraphics[page=1,width=0.5\textwidth,height = 0.25\textheight]{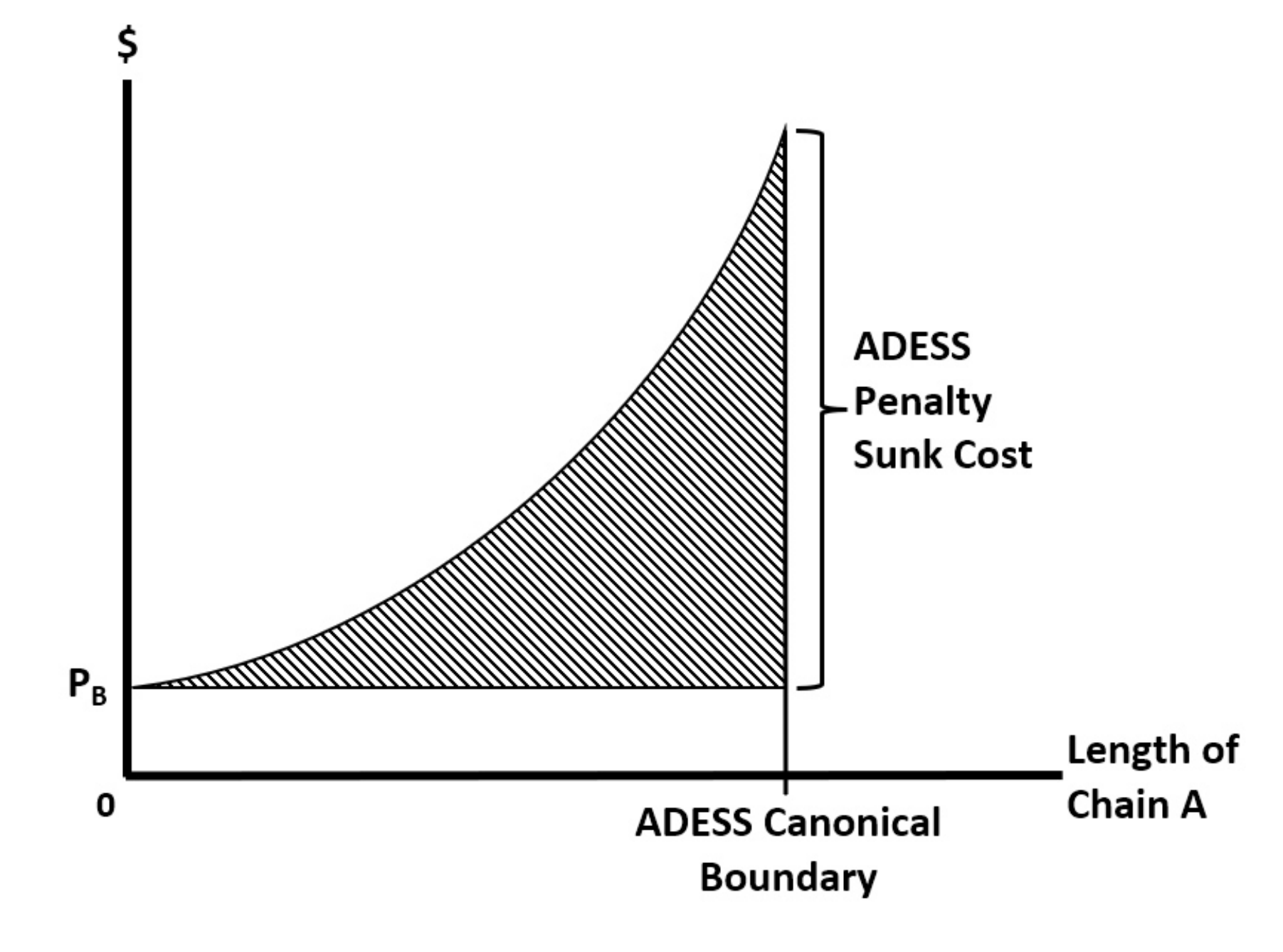}
\end{center}
\caption{The ex-ante expected exponentially increasing sunk cost of a double-spend attack}
\label{fig: The exponential sunk cost of a double-spend attack}
\end{figure}

\subsection{The intrinsic deterrent to double-spend attacks}
\label{subsec: The intrinsic deterrent to double-spend attacks}

The Penalty Assignment Rule and the Penalty Function ensure that, to succeed,  a double-spend attacker must use more hashrate under $ADESS$ compared to Nakamoto. This is a direct result of the penalty and stands independent of any strategic considerations or models. The key implementation issue is the choice $\alpha$. We do not, at present, know the lower bound of confirmation blocks. A goal of future research is to attempt to estimate $\alpha$. Until then, it may be reasonable to use an existing convention, such as Nakamoto's six-block confirmation rule for Bitcoin. It is also possible that, once chosen, $\alpha$ will become a Schelling focal point whereby recipient nodes wait at least $\alpha$ confirmation blocks  before sending their exchange item.\footnote{Schelling (1960) \cite{Schelling (1960)}.} Proposition 1 states the relative increase in hashrate required to carry out a double-spend attack under $ADESS$.

\begin{proposition}[The Increased Cost of Attack Under $ADESS$]
The ex-ante expected hashrate required to carry out a double-spend attack under $ADESS$, where the penalty is bounded away from zero, i.e. $\xi> e$ for some $e > 0$, is weakly greater than the ex-ante expected hashrate required to carry out a double-spend attack under Nakamoto. 
\end{proposition}

\begin{proof}
Consider any post-fork block $N$ on chain $\mathcal{IC}$ at which chain $\mathcal{A}$ becomes canonical under $ADESS$. The expected hashrate is $N(1 + \xi)$ if $N <\alpha$ or there is no mining puzzle adjustment at any block. A mining puzzle adjustment at any block increases the expected hashrate for succeeding blocks by a exponential factor. Under Nakamoto, the attack strategy from Section \ref{subsec: Nakamoto's Claim and Budish's Critique} to become canonical at block $N$ on chain $\mathcal{IC}$ is $N + e$. Since $\xi > e$, it follows that the expected hashrate under $ADESS$ always exceeds the expected hashrate under Nakamoto. The proposition follows from the fact that hashrate costs money.
\end{proof}

\section{A Model of $ADESS$}
\label{sec: The Model}

In this Section, we formalize a model of the $ADESS$ protocol and evaluate its resilience to double-spend attacks under a set of baseline conditions. The model applies to an attacker who is deciding whether to launch a double-spend attack on a single node to which is has sent tokens in a single transaction. It evaluates the attacker's  'go', 'no go' decision to launch an attack, based on the ex-ante expected profit from an optimal attack strategy. It does not address any change in the attacker's plan after the attack is launched.\footnote{The model derives an upper bound to double-spend vulnerability, since an attack, once launched, may be discontinued before it reaches completion.} We make two assumptions that simplify the analysis and enable us to focus on how the $ADESS$ protocol affects the decision to launch a double-spend attack. The first assumption limits the strategic space in the  game played by miners. The second assumption imposes a deterministic chain growth rate. 

\subsection{Mining game structure}
\label{subsec: Game Structure}

There is a single attacker who attempts to double-spend a single counterparty.  The attacker mines in secret until it has achieved two milestones; chain $\mathcal{A}$ has crossed the canonical boundary and the attacker has received the exchange item from its victim, whereupon the attacker broadcasts chain $\mathcal{A}$ to the network. There are multiple honest miners of chain $\mathcal{IC}$. These miners do not engage in double-spend attacks. They only become aware of a double-spend attack when the attacker broadcasts its chain. If chain $A$ meets the criteria to be canonical, the honest miners switch from mining chain $\mathcal{IC}$ to mining chain $\mathcal{A}$. Honest miners may (or may not) have rational expectations about the possibility of attack and may require a premium profit to compensate for the risk of losing their escrowed  post-fork block rewards on chain $\mathcal{IC}$ if an attack occurs and chain $\mathcal{A}$ becomes canonical. The model is agnostic as to the competitiveness of the mining market. The attacker is the only strategic player in the game we analyze. We characterize the equilibrium by evaluating the attacker's decision to attack. 

\begin{assumption*}[Strategic Space]
\label{ass: Strategic Space}
The attacker makes a strategic choice of whether and when to launch a double-spend attack. Honest miners work on the canonical chain.\qed
\end{assumption*}

\subsection{Certainty equivalence for blockchain growth rate}

Miners make a succession of independent guesses of puzzle solutions over time. The time distribution of new blocks is a Bernoulli Process parameterized by the hashrate and puzzle difficulty, with the mean reverting to the target growth rate by periodic adjustments to puzzle difficulty. This random process can create opportunity for strategic behavior at any point in time, for example when a miner solves a puzzle faster than the target, as analysed by Eyal and Sirer (2014) \cite{Eyal and Sirer (2014)}. However, the decision to launch an attack, which $ADESS$ is designed to influence, is based on, inter alia, the ex-ante expected growth rate of the blockchain, which is a deterministic function of hashrate and puzzle difficulty. In order to focus on the decision to launch a double-spend attack, our model uses the certainty equivalence of the ex-ante expected growth rate.

\begin{assumption**}[Certainty Equivalence of Mining Puzzle Solutions]
\label{Certainty Equivalence of Mining Puzzle Solutions}
Blockchain growth is a deterministic function of hashrate and mining puzzle difficulty. \qed
\end{assumption**}

\subsection{The baseline case}
\label{sec: The Baseline Case}

Initially we impose four restrictions on the model. The first two are relaxed in Section \ref{sec: Relaxing the Baseline Model Constraints} and the latter two are relaxed in the Appendix.\vskip5pt 

(i) After each new block, the mining puzzle difficulty fully adjusts - to set the growth rate at 1 block per unit of time at the hashrate applied to the previous block.\vskip5pt

(ii) All nodes observe blocks in the same temporal order.\vskip5pt

(iii) There is one attacker chain $\mathcal{A}$ and one victim.\vskip5pt

(iv) The attacker chooses a constant growth rate $\gamma$ for chain $\mathcal{A}$.\vskip5pt

In order to reduce the number of variables in the analysis we impose several additional restrictions throughout which do not affect the implications of the model. The growth rate target is 1 block per unit of time; hashrate on chain $\mathcal{IC}$ is 1 per unit of time; the block reward, $p_{B}$, is constant; the the attacker faces the same cost of 1 unit of hashrate, $c$ -  in dollars - as the honest miners; the profit from honest mining is normalized to zero (i.e. $P_{B} = c$)\footnote{The normalization reflects that the attacker's profit from mining is equal to the market return on investment. This may, or may not, imply a restriction on mining market structure. For example, if a mining oligopoly resulted from active miners having a lower cost of hashrate compared to other miners, the attacker would not earn an extra-normal return from its mining. On the other hand, a low cost miner does not need to create a fork (i.e. launch an attack) in order to earn an extra-normal profit. It can earn the profit so by mining chain $\mathcal{IC}$.}; the crypto to dollar exchange rate at 1 and the time discount rate at $\delta \in (0,1]$. Finally, we assume the attacker has a consistent estimate of the victim's confirmation depth, which may exceed $\alpha$ as defined in Section \ref{sec: The $ADESS$ Protocol} based on, inter alia, the transaction value or the identity of, or its relationship with, the sender. In a slight abuse of notation, we use $\alpha$ to denote the confirmation depth in the model of this section.\footnote{Confirmation depth $\alpha$ is an equilibrium outcome of a game between the counterparties that we do not model.} 

\subsubsection{The attacker's strategic plan}
\label{subsec: The attacker's decision problem}

We decompose the attacker's strategic plan into three parts. The first decision is the choice of which block to fork and when to begin building chain $\mathcal{A}$. The second decision is the block on $\mathcal{IC}$ at which chain $\mathcal{A}$ reaches the canonical boundary. The third decision is when to broadcast chain $\mathcal{A}$. We evaluate each decision. 

As to the first decision, the attacker can start to build chain $\mathcal{A}$ by forking a block on the canonical chain at any time. The optimal choice of starting time is the solution of a dynamic optimization problem. A complicating factor in the analysis is that the wait time for the transaction to appear on a block is stochastic. There are two levels of uncertainty. One is the probability of a transaction being chosen by a miner, which can be affected inter-alia by the tip fee offered by the attacker relative to other tip fees in the transaction pool. The other is the probability of a miner being first to solve the mining puzzle. We denote the the number of blocks on chain $\mathcal{IC}$ between the fork-block and the block onto which the transaction is appended by the random variable $\sigma$.

We do not solve the full dynamic programming problem, however, we are able to show that the attacker will fork  the head of the canonical chain. Suppose the attacker forks a block that is $\tau$ blocks earlier than the head  of chain $\mathcal{IC}$. In order for chain $\mathcal{A}$ to cross the canonical boundary when chain $\mathcal{IC}$ has $N$ post-fork blocks, it needs to append $N(1 + \xi)$ blocks in $N - \tau$ units of time. Denote $\gamma$ as the growth rate of the attacker chain required to reach the canonical boundary at that time. The optimal $\tau$ is derived by solving Equation \ref{eq: Timing of commencement of attack} to minimize the hashrate required to cross the canonical boundary. 

\begin{equation}
\label{eq: Timing of commencement of attack}
\underset{\text{chain $\mathcal{A}$ blocks}}{\underbrace{N(1 + \xi)}} =  N\underset{\text{growth rate to reach canonical}}{\underbrace{(1 + \gamma)}} - \tau
\end{equation}

Equation \ref{eq: Timing of commencement of attack} implies that setting $\tau = 0$ minimizes cost (since the growth rate is a function of hashrate which has a unit cost $c$), which means the attacker will fork an end-block. This solution, in turn, implies that the attacker will choose the growth rate $(1 + \xi)$, which matches the penalty.\footnote{The problem can be stated as $\arg\min\limits_{\tau}\gamma = N(1 + \xi) + \tau + 1$.} We do not evaluate the decision of which head-block of the canonical chain the attacker will commence building chain $\mathcal{A}$. Rather, we analyse the attacker's decisions contingent on having forked a blockchain head.

As to the second decision, Equation \ref{eq:$A$'s problem-length of attack} is the attacker's problem for choosing the length of the attack, which is the block $N$ on chain $\mathcal{IC}$ at which it reaches the canonical boundary.\footnote{The exponent of the discount rate is divided by the growth rate to adjust for the intervals of time between blocks on chain $\mathcal{A}$ prior to reaching the canonical boundary, which is less than 1. The bracketed expression  $\lceil N(1+\xi)-1\rceil$ indicates that the number of blocks on chain $\mathcal{IC}$ is rounded up from the value inside the brackets.} The objective is to minimize cost. 

\begin{equation}
\label{eq:$A$'s problem-length of attack}
\arg\min\limits_{N}\; \text{cost of attack}\; = c\sum_{n = 0}^{\lceil N(1 + \xi) - 1\rceil} \delta^{n/(1+\xi)}(\;\underset{\text{Growth rate}}{\underbrace{ 1 + \xi}})^{n}
\end{equation} 

{\noindent{Subject to $N \geq \sigma + \alpha$, which is the minimum interval over which the penalty is applied. It is immediate that cost is increasing in $N$. We conclude that chain $\mathcal{A}$ will reach the canonical boundary at exactly $N = \alpha + \sigma$ post-fork blocks on chain $\mathcal{IC}$. }}

The attacker's third decision is the time to broadcast chain $\mathcal{A}$. The question is whether the attacker will continue to mine chain $\mathcal{A}$ in secret past the time it reaches the canonical boundary. Let $B$ denote the blocks that the attacker secretly mines after chain $\mathcal{A}$ has reached the canonical boundary.  The attacker will have to apply at least 1 unit of hashrate on chain $\mathcal{A}$ on each of the $B$ blocks in order to remain canonical when it is broadcast. This follows from the Penalty Function, which states that once it has reached the canonical boundary, chain $\mathcal{A}$'s score is adjusted to the cumulative mining puzzle difficulty on chain $\mathcal{IC}$ plus a small $\epsilon$. Chain $\mathcal{A}$ needs to keep pace with chain $\mathcal{IC}$ to ensure it will be canonical when it is broadcast. Equation \ref{eq: Attacker's choice of $B$} is the problem the attacker solves to determine when to broadcast its chain.\footnote{In an abuse of notation, $\lambda$ in the attacker's ex-ante decision problem denotes the mean of the distribution of the random variable $\sigma$.}

\begin{multline}
\label{eq: Attacker's choice of $B$}
\arg\max\limits_{B}\pi(B;\xi,N,v, p_{B}) =  \underset{\text{discounted revenue}}{\underbrace{\delta^{(\alpha + \sigma) + B - 1} \{v + p_{B}(\lceil N(1 + \xi)\rceil + B)\}}}\\
- c\underset{\text{discounted cost}}{\underbrace{\big[\sum_{n =0}^{\lceil (\alpha + \sigma)(1 + \xi) - 1\rceil} \{\delta^{n/(1+\xi)}( 1 + \xi)^{n}\} + \sum_{b=0}^{B - 1}\delta^{(\alpha + \sigma) + b}\big]}}
\end{multline} 

{\noindent{Revenue is discounted by the time elapsed from the commencement of the attack to the broadcast of the attack, reflecting that the attacker is not able to spend the block rewards from chain $\mathcal{A}$ until it is broadcast. Costs are discounted from the time they are incurred. The discounted profit from secretly mining one block after the time at which the canonical boundary is crossed (equivalently, past $\alpha + \lambda$ blocks on chain $\mathcal{A}$)is}}

\begin{multline}
\label{eq: marginal cost past the canonical boundary}
M(B) = \underset{\text{discount of mining rewards on ancestor blocks}}{\underbrace{(\delta^{(\alpha + \sigma)} - \delta^{(\alpha + \sigma -1)})\{v + p_{B}(\lceil N(1 + \xi)\rceil + 1)\}}}\\
+ \underset{\text{disc rev from added block}}{\underbrace{\delta^{(\alpha + \sigma)}p_{B}}} - \underset{\text{disc cost of added block}}{\underbrace{\delta^{(\alpha + \sigma)}c}}
\end{multline}

{\noindent{Rearranging the terms of Equation \ref{eq: marginal cost past the canonical boundary} into Equation \ref{eq: Attacker's broadcast decision} shows that the marginal profit from secretly mining past the canonical boundary is negative. It follows that the attacker will broadcast chain $\mathcal{A}$ at $\alpha + \sigma$ post-fork blocks on chain $\mathcal{IC}$.}}
 
\begin{equation}
\label{eq: Attacker's broadcast decision}
M(B) = \underset{< 0}{\underbrace{(\delta^{(\alpha  + \sigma)} - \delta^{(\alpha + \sigma -1) })}}\underset{> 0}{\underbrace{\{v + p_{B}(\lceil N(1 + \xi)\rceil + 1)\}}} + \underset{= 0}{\underbrace{\delta^{(\alpha + \sigma)}(p_{B} - c)}} \leq 0
\end{equation}

{\noindent{We conclude that the attacker's optimal strategic plan is to fork a head-block and grow chain $\mathcal{A}$ at the rate of $(1 + \xi)$ (the first decision); to reach the canonical boundary when chain $\mathcal{IC}$ has $\alpha + \sigma$ post-fork blocks (the second decision) and thereupon to broadcast chain $\mathcal{A}$ (the third decision). If the expected profit from this plan is positive, the attack will be launched.\footnote{This is not a complete characterization of the attacker's decision problem since it  does not pin down the block that the attacker will fork. The conclusions are conditional on $\lambda$. However, this imprecision turns out not to matter for the results we are interested in.}}}

\section{Key Properties of the $ADESS$ Protocol}
\label{sec: Some Properties of the $ADESS$ Protocol}

In this section we state the main result of our model, which is that the penalty parameter $\xi$ can be set to render unprofitable a double-spend attack on a transaction of any size. Equation \ref{eq: Attacker's choice of $B$} is the attacker's ex-ante expected profit. Noting that $N$, $p_{B}$ and $c$ are fixed and $B = 0$,  profitability is determined by the interaction of the value of the transaction $v$ and the penalty parameter $\xi$. For a given $v$, an increase in $\xi$ raises discounted revenue and cost. Equation \ref{eq: marginal condition of an attack} is the the marginal effect of an increase in $\xi$, by $d\xi$,  on the attacker's profit, where $N$ is the interval of blocks on chain $\mathcal{IC}$ between the fork block and the broadcast of chain $\mathcal{A}$ and the marginal increase in $\xi$ causes the number blocks on chain $\mathcal{A}$ required to reach the canonical boundary to increase by one block (i.e. $\lceil N(1 + (\xi + d\xi)) \rceil = \lceil N(1 + \xi)\rceil$ +1).

\begin{multline}
\label{eq: marginal condition of an attack}
M(\xi) = \underset{\text{additional block reward}}{\underbrace{\delta^{N}p_{B}}} - c\underset{\text{cost increase on current blocks}}{\underbrace{\sum_{n = 0}^{\lceil N(1 + \xi) -1\rceil}d/d\xi(\delta^{n/(1+\xi)}(1+\xi)^{n})}}\\
- \underset{\text{cost of additional block}}{\underbrace{\delta^{N}c}}
\end{multline}
The derivative for the discounted cost of an attack at $n$th block on chain $\mathcal{A}$ is
\[d/d\xi(\delta^{n/(1+\xi)}(1+\xi)^{n}) =  n(\xi + 1)^{n-2}\delta^{n/(\xi + 1)}[(\xi +1) - log(\delta)]\]

{\noindent{This expression is positive since the time discount rate $\delta \in(0,1)$, which implies that  $log(\delta) < 0$. Under the assumption that $p_{B} = c$, it is the case that $M(\xi) < 0$ for any $\xi > 0$, i.e. that the ex-ante expected profit from an attack is a declining function of the penalty parameter $\xi$.\footnote{If $d\xi$ does not cause the number blocks on chain $\mathcal{A}$ required to reach the canonical boundary to increase, the leftmost and rightmost terms of Equation \ref{eq: marginal condition of an attack} drop out and $M(\xi)$ remains negative.} The key property of $ADESS$ is that, for every transaction value $v$ there is a penalty parameter $\xi$ above which attack is unprofitable. We state this property in the following theorem.}}

\begin{theorem*}[$ADESS$ Bound on Double-Spend Profitability]
\label{thm: Bound on Double-Spend Attacks}
For any transaction value $v$, there is a  value of the penalty parameter $\underline{\xi}$ above which a double-spend attack is unprofitable. 
\end{theorem*}

\begin{proof}
Equation \ref{eq: marginal condition of an attack} shows that ex-ante expected profit from a double-spend attack is a decreasing function of the penalty parameter $\xi$. The derivative of $M(\xi)$ with respect to $\xi$ is

\begin{multline}
\label{eq: second derivative}
\frac{d}{d\xi}M(\xi) = - c\sum_{n=0}^{\lceil N(1+\xi)-1\rceil}d^{2}/d\xi^{2}(\delta^{n/(1+\xi)}(1+\xi)^{n}) < 0 
\end{multline}
The second derivative for the discounted cost of an attack at $n$th block on chain $\mathcal{A}$ is

\begin{multline}{\nonumber}
d^{2}/d\xi^{2}(\delta^{n/(1+\xi)}(1+\xi)^{n}) = \\
n(\xi + 1)^{n-4}\delta^{n/(\xi + 1)}[n -(n-1)(\xi + 1)]2log(\delta) + (n-1)(\xi + 1)^{2}
\end{multline}

{\noindent{This expression is positive when the term $n -(n-1)(\xi + 1) <0$, which holds for $\xi > 1/(n-1)$. Equation \ref{eq: second derivative} applies whether or not $d\xi$ induces an additional block to be added to chain $\mathcal{A}$ before reaching the canonical boundary. $M(\xi) > 0$ and $\frac{d}{d\xi}M(\xi) > 0$, imply that, for any transaction value $v$, there is a value of $\xi$ above which the expected profit from a double-spend attack is a strictly concave  decreasing function of the penalty parameter, with a discontinuous drop in profit where an additional block is added. It follows that there is a value of $\xi$ for which $\pi(\xi) < 0$.\footnote{The proof is as follows. By strict concavity, for $\xi' > \xi" \; \frac{\pi(\xi') - \pi(0)}{\xi'} \leq \frac{\pi(\xi") - \pi(o)}{\xi"}\implies \pi(\xi') \leq (\xi'(\pi(\xi") - \pi(0)) + \pi(0)$. Since $\pi(\cdot)$ is decreasing in $\xi$, profit is negative for any $\xi > \xi'$.} This proves the Theorem. Figure \ref{fig: The Attacker's Profit Function} displays the relationship between $\xi$ and the attacker's profit.}}

\begin{figure}
\begin{center}
\includegraphics[page=1,width=0.5\textwidth,height = 0.3
\textheight]{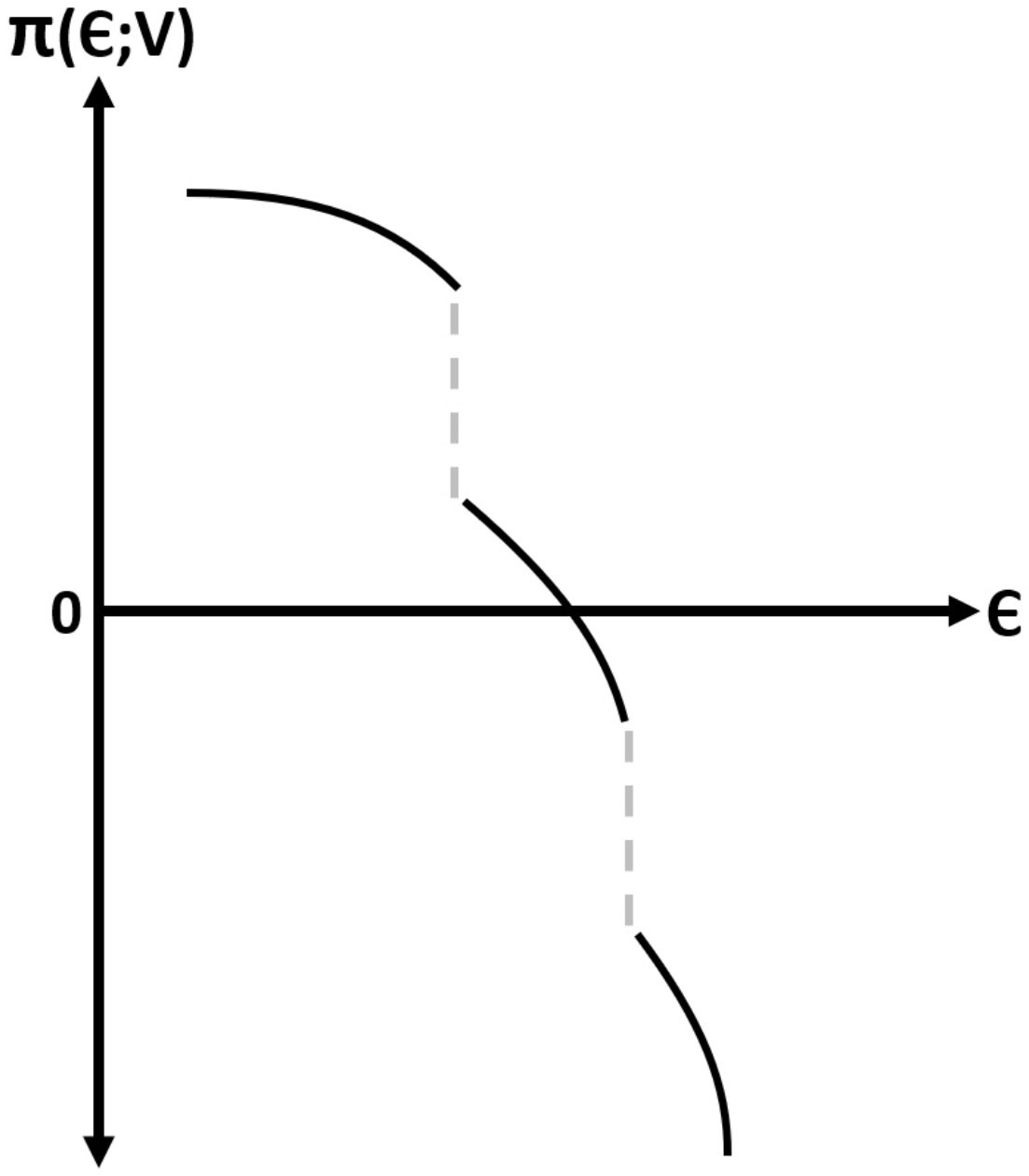}
\end{center}
\caption{The Attacker's Profit Function}
\label{fig: The Attacker's Profit Function}
\end{figure}
\end{proof} 

$ADESS$ has the analogous property that every penalty setting renders unprofitable some interval of transaction values.

\begin{corollary*}
For any penalty parameter $\xi > 0$ there corresponds an interval of transaction values $[0,v)$ for which a double-spend attack is unprofitable.
\end{corollary*}

\begin{proof}
Set $v=0$ in Equation \ref {eq: Attacker's choice of $B$}. Under our assumption that $p_{B} \leq c$, it is immediate that $\pi(\xi,v = 0) < 0$ for any $\xi >0$. It follows that $\pi(\xi,v) < 0$ for any $v\in [0,- \pi(\xi,v = 0)$. This proves the Corollary.
\end{proof}

\section{Relaxing the baseline model constraints}
\label{sec: Relaxing the Baseline Model Constraints}

In this section we relax two of the restrictive assumptions in the baseline model of section \ref{sec: The Baseline Case} and evaluate how the key properties of the $ADESS$ protocol are affected in each instance.

\subsection{Incomplete or infrequent adjustment of mining puzzle difficulty}

There are two cases to consider here. The first case occurs when the puzzle difficulty does not fully adjust after each block (but it does adjust symmetrically). Consider the case where the hashrate applied to the first block is $(1 +\gamma)$. Under Assumption 2 the puzzle is solved in $1/(1 + \gamma)$ units of time. Full adjustment implies the puzzle difficulty is increased to require a hashrate of  $(1 + \gamma)$ to solve the puzzle for block \#2. With a partial adjustment factor $\beta \in (0,1]$, the required hashrate is $(1 + \beta\gamma)$, which a lesser increase than full adjustment.\footnote{In the context of footnote 17, partial adjustment implies that an increase of hashrate to to $1+ \gamma$ causes difficulty to increase to $D = 1 + \beta\gamma$, and so forth.} The hashrate growth of a sequence of blocks is the sum $(1 + \gamma) + (1 +\beta\gamma)(1 + \beta\gamma)+...+ (1+\gamma)(1 + \beta\gamma)^{n-1} + ...$. $\beta$ slows the exponential growth of the attacker's cost. The only effect this partial adjustment has on the $ADESS$ protocol is to reduce the effective penalty by the mapping $\xi \implies \beta\xi$. The effective penalty can still be set at any value.

The second case occurs when the mining puzzle adjustment is made periodically, after an interval of blocks has been appended to a chain (called an "epoch"). For example, in Bitcoin the puzzle difficulty is adjusted every 2016 blocks. Inside of an epoch, the cost of achieving the penalty growth rate grows linearly rather than exponentially. The increased cost of growing chain $\mathcal{A}$ at the rate of $(1 + \xi)$ is an additional $c\xi$ per block on chain $\mathcal{IC}$. If the block reward equals cost, this cost is recoverable from block rewards. The conclusion is that, for a given penalty parameter $\xi$, the effect of the penalty on attack cost increases with the frequency of mining puzzle adjustments.

\subsection{Network latency} 
\label{subsec: Network latency}

Consider a synchronous network with a block propagation delay upper bound of $\Delta > 0$ units of time and no other channel of communication between nodes. In this setting it is possible that some nodes receive the broadcast of post-fork block $\alpha + \sigma$ on chain $\mathcal{IC}$ before it receives the broadcast of post-fork block $\lceil N(1 + \xi)\rceil$ on chain $\mathcal{A}$, even when $\mathcal{A}$ was the first to initiate a broadcast. These nodes will not recognize chain $\mathcal{A}$ as canonical. Under $ADESS$, a split in opinion at the canonical boundary will persist if the following two conditions hold: (i) a portion of honest miners do not observe that chain $\mathcal{A}$ has reached the canonical boundary and continue to mine chain $\mathcal{IC}$ and the other portion mine chain $\mathcal{A}$ and (ii) the hashrate applied to mining chain $\mathcal{IC}$ is weakly below the hashrate applied to chain $\mathcal{A}$. In this case, there is more hashrate applied to chain $\mathcal{A}$ than chain $\mathcal{IC}$, but it is a constant amount. As a result, miners of chain $\mathcal{A}$ will observe that their chain remains canonical since it has more post-crossing cumulative mining puzzle difficulty than chain $\mathcal{IC}$, and miners of chain $\mathcal{IC}$ will observe their chain remains canonical since cumulative mining puzzle difficulty on chain $\mathcal{A}$ is not growing exponentially while it remains penalized. 

The attacker can prevent a breakdown of consensus at the canonical boundary by broadcasting $(N + \Delta)(1+\xi)$ blocks on chain $\mathcal{A}$ at the time $N$ post-fork blocks on chain $\mathcal{IC}$ are broadcast. This requires chain $\mathcal{A}$ to grow at the accelerated rate of $1 + \{\xi + \frac{\Delta}{N}(1 + \xi)\}$. A node that receives chain $\mathcal{A}$'s broadcast $\Delta$ units of time in the future while receiving the broadcast of chain $\mathcal{IC}$ without delay will observe that chain $\mathcal{A}$ reached the canonical boundary at the time chain $\mathcal{IC}$ has $N + \Delta$ post-fork blocks. Theorem 1 and Corollary 1 continue to hold when the growth rate $\xi$ replaced by the accelerated growth rate.\footnote{Note that this result holds when latency slows the growth rate of chain $\mathcal{IC}$ due to conflicts arisng from uncle chains.}

\section{Blockchain Security of \textit{ADESS} vs. Nakamoto}
\label{sec: Additional Considerations}

In this section we compare $ADESS$ to Nakamoto with respect to two matters not previously covered; a malicious attack and an uninformed node.  Our conclusion is that, in these two matters, $ADESS$ does not introduce any serious security vulnerabilities that are not present in Nakamoto. 

\subsection{Malicious attack} 
\label{subsec: Malicious Attack}
We define a malicious attack as an attempt to break the consensus around one canonical chain. In Section \ref{subsec: Network latency} we showed that a permanent split in the consensus opinion under $ADESS$ could arise when chain $\mathcal{A}$ reaches the canonical boundary if the crossing is not seen by some honest miners. In the event of a split where $\leq 50\%$ of honest hashrate was applied to chain $\mathcal{IC}$, a malicious attacker could withdraw its hashrate and the split would be permanent. To achieve this result the attacker must incur (a) the exponentially increasing cost of reaching the canonical boundary and (b) the cost of eclipsing some miner nodes around the canonical boundary to ensure that the requisite number do not see chain $\mathcal{A}$ reach the canonical boundary. Once that has been achieved, honest miners will work on both chains under the assumptions of our model. More generally, consensus cannot be re-established under the $ADESS$ protocol. 

A permanent split in the consensus opinion under Nakamoto can be achieved if the attacker either (i) indefinitely applies hashrate to ensure that chain $\mathcal{IC}$ and the attacker's chain $\mathcal{A}$ grow at the same rate or (ii) the attacker permanently eclipses miners representing $< 50\%$ of hashrate so as to prevent them from seeing chain $\mathcal{IC}$ while broadcasting chain $\mathcal{A}$ to them. In that case, the eclipsed miners will mine chain $\mathcal{A}$ and the un-eclipsed miners will apply their hashrate to chain $\mathcal{IC}$, which they observe to have the most cumulative mining puzzle difficulty. 

This suggests a tradeoff between $ADESS$ and Nakamoto. Under $ADESS$ a malicious attacker must overcome the penalty to carry out the attack, which implies a larger initial investment to launch the attack compared to Nakamoto, but no long-term expenditures. On the other hand, Nakamoto requires a perpetual application of computing power to maintain the consensus split. Figure \ref{fig: ADESS v Nakamoto attack} is a graphical display of the relative costs of a malicious attack under each protocol. If costs are discounted at the rate of $\delta\in (0,1)$ per unit of time, the early costs incurred under $ADESS$ will be weighted more heavily than the further out costs incurred under Nakamoto. In that case, it is possible that the present value of the $ADESS$ attack cost could be higher than the present value of the Nakamoto attack cost, even when the actual $ADESS$ cost is below the Nakamoto cost. There is no apriori way to rank the difference in the present value cost of attack between the two protocols. 

\begin{figure}
\begin{center}
\includegraphics[page=1,width=0.65\textwidth,height = 0.3\textheight]{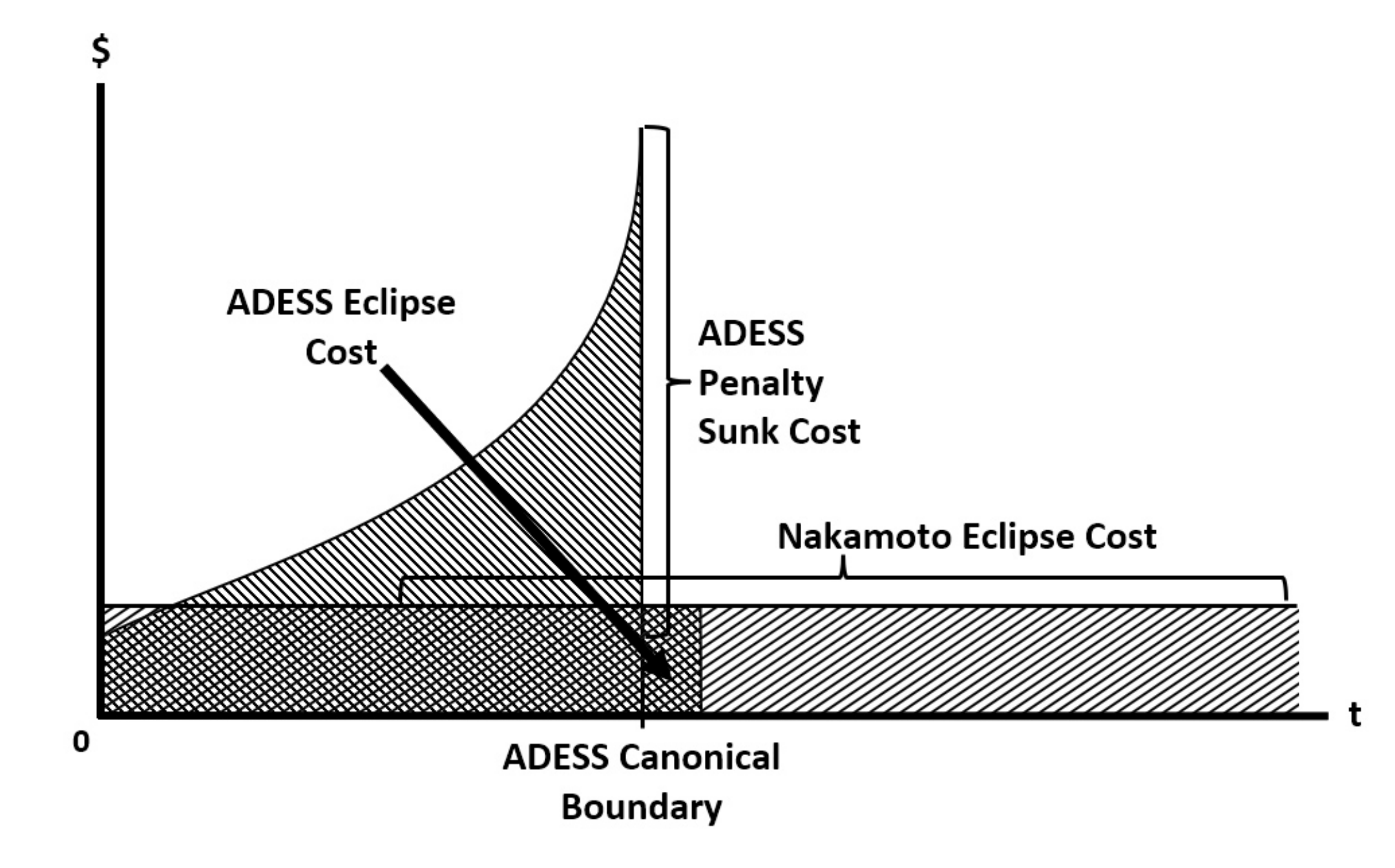}
\end{center}
\caption{\textit{ADESS} v Nakamoto - cost of malicious attack}
\label{fig: ADESS v Nakamoto attack}
\end{figure}

\subsection{A node is not connected when a fork occurs}

A node that is not connected to the network when the fork occurs does not observe the temporal order of the broadcast of the initial blocks of the post-fork chain segments. Under $ADESS$, such a node cannot make a determination of whether, or to which chain, a penalty should be assessed and therefore cannot determine which chain is canonical. By contrast, Nakamoto does not require a node to observe the temporal sequence of new blocks. A node that newly enters the network can determine the canonical chain by comparing cumulative mining puzzle difficulty. This is a weakness in the security of $ADESS$ compared to Nakamoto. Nevertheless, it is unclear what practical effect this weakness represents. $ADESS$ provides an incentive for miners and frequent transactors, such as exchanges, to maintain several nodes - which are inexpensive - in order to ensure it is always connected to the network. Less frequent participants in the network can infer the canonical chain by observing which chain is being actively mined.

\section{Conclusion}

We reviewed the literature on the vulnerability of PoW blockchains to double-spend attacks that is relevant to our analysis. That literature shows that such attacks can be profitable at any transaction size and that the possibility of retaliation by the victim is not necessarily an effective deterrent. We then proposed a modification to the standard PoW protocol, called $ADESS$, that increases ex-ante expected  cost of carrying out a double-spend attack by assigning a penalty to the scoring of the attacker's fork chain. There are two parts to $ADESS$. The first part is a criteria for identifying the attacker's chain and assigning the penalty to it. We argued that identification is made possible by the fact that a rational counterparty will convey its item of value only after the the block to which the transfer transaction is appended and several descendant confirmation blocks have been added to the canonical chain. A rational attacker will not broadcast its chain until after it has received the item of value from its counterparty, which requires it to wait until after the confirmation blocks have appeared before broadcasting its chain. The resultant delay in broadcasting the attacker's chain suggests the criteria of assigning the penalty to the chain that does not broadcast its blocks until after the other chain has added several post-fork blocks. 

The second part of $ADESS$ is the penalty function. Once the penalty has been assigned, the criteria for comparing chains shifts from cumulative mining puzzle difficulty to the number of blocks. The penalty discounts the value of block on the penalized chain. The penalized chain must therefore grow at a faster rate to overcome the penalty. The interaction of faster growth with mining puzzle difficulty adjustments that require ever higher hahsrate to maintain an elevated blockchain growth rate, result in a exponential increase in hashrate - and cost - for the attacker's chain to overcome the penalty. We showed that the expected cost of carrying out a double-spend attack under $ADESS$ is weakly higher than Nokamoto. 

We constructed a model which enabled us to prove that, for a transaction of any size, there is an $ADESS$ penalty that renders a double-spend attack unprofitable. We then demonstrated that the main results continue to hold in the presence of incomplete adjustment of mining puzzle difficulty and network latency. Finally, we argued that the requirement that nodes observe the temporal sequence of blocks in $ADESS$ does not practically create a serious impediment to the participation of nodes in the network and that there is no apriori way to rank the vulnerability to a malicious attack between $ADESS$ and Nakamoto.

\section*{Acknowledgments}
We thank Glenn Ellison and Rainer B\"ohme for thoughtful comments and guidance. All errors are our own. \\

\begin{appendix}

\section{Generalizing $ADESS$ to multiple chains and multiple forks}
\label{app: Generalizing $ADESS$ to multiple chains and multiple forks}

$ADESS$ can be extended to a general blockchain network in which there are multiple forks, each with two or more descendant chains where penalty assignments are made at each fork. To accommodate this, we generalize the $ADESS$ protocol so that a penalized chain cannot become canonical until it has overcome all penalties assigned to it.  Proposition 1 and its Corollary prove that there is always exactly one canonical chain.

\subsection{Tree graph representation of a blockchain network}

A blockchain network can be represented as a directed tree graph. A chain is a unique directed path running from a fork-block to a chain head.\footnote{Unique paths, running from the root node to end nodes, are a feature of directed tree graphs.} At any time $t$ there are $N$ fork-blocks in the blockchain network, each one denoted $f_{n},\; n \in \{1,...,n,...,N\}$ and $M$ heads, each one denoted $B_{m},\; m \in \{1,...,m,...,M\}$. We also denote $B_{m}$ as the chain connecting a fork-block to a chain head. Figure \ref{fig: A Blockchain Network} displays the blockchain network as a directed tree graph at $t$ periods after the Genesis block was broadcast.

\subsection{Penalty nomenclature}

We simplify the presentation by assuming that an attacker forks the parent of the block to which the transaction is appended, so that $\lambda = 0$.\footnote{From Section \ref{subsec: The attacker's decision problem}, $\lambda$ is the number of blocks in chain $IC$ (in the generalized context the baseline chain $B_{m'}$) between the fork-block and the block onto which the transaction is appended. $\lambda > 0$ at fork $n$ would be denoted $\alpha + \lambda(f_{n}):f_{n}$, which would not affect Proposition 1.} The baseline chain for fork-block $f_{n}$ is the first chain to broadcast $\alpha$ post-fork blocks, which is denoted $\alpha:f_{n}$. For example in Figure \ref{fig: A Blockchain Network}, $\alpha:f_{2}$ indicates that chain $B_{4}$ is the baseline chain relative to $f_{2}$, since it was the first chain to broadcast a chain segment with $\alpha$ blocks, starting at $f_{2}$. A more complicated example involves chain $B_{1}$. $\alpha:f_{1}$ indicates that $B_{1}$ was the first chain to broadcast a chain segment with $\alpha$ blocks, starting at $f_{1}$ and  $\alpha:f_{3}$ indicates that chain $B_{2}$ was the first chain to broadcast a chain segment with $\alpha$ blocks, starting at $f_{3}$. In this case $B_{1}$ is the baseline chain for fork $f_{1}$ and is not the baseline chain for fork $f_{3}$. At each fork $f_{n}$ a penalty may (or may not)  be assigned to one or more chains $B_{m}$. When a assignment is made, the penalized chain is compared to the baseline chain $B_{m'}$.\footnote{A baseline chain at one fork-block can be a penalized chain at another fork-block.}

The penalty is "active" at fork-block $f_{n}$ so long as the penalized chain has not overcome the penalty .  The tuple $<B_{m}:\widehat{f}_{n}:B_{m'}:t>$ denotes that $B_{m}$ is actively penalized relative to $B_{m'}$ at $f_{n}$ at time $t$. The penalty is "inactive" after it has been overcome by the penalized chain. An inactive penalty is denoted $<B_{m}:f_{n}:B_{m'},t>$. In the latter case, $B_{m}$ has become the baseline chain relative to fork $f_{n}$. In general, $\widehat{f}_{n}$ indicates a currently active penalty and $f_{n}$ indicates a past penalty that has been overcome. The set of active and inactive penalties assigned to $B_{m}$ at time $t$ is denoted by a list such as $\{<B_{m}:\widehat{f}_{n}:B_{m'}:t>, <B_{m}:f_{n+i}:B_{m'+j}:t>, ...\}$.

\begin{figure}
\begin{center}
\includegraphics[page=1,width=0.55\textwidth,height = 0.27
\textheight]{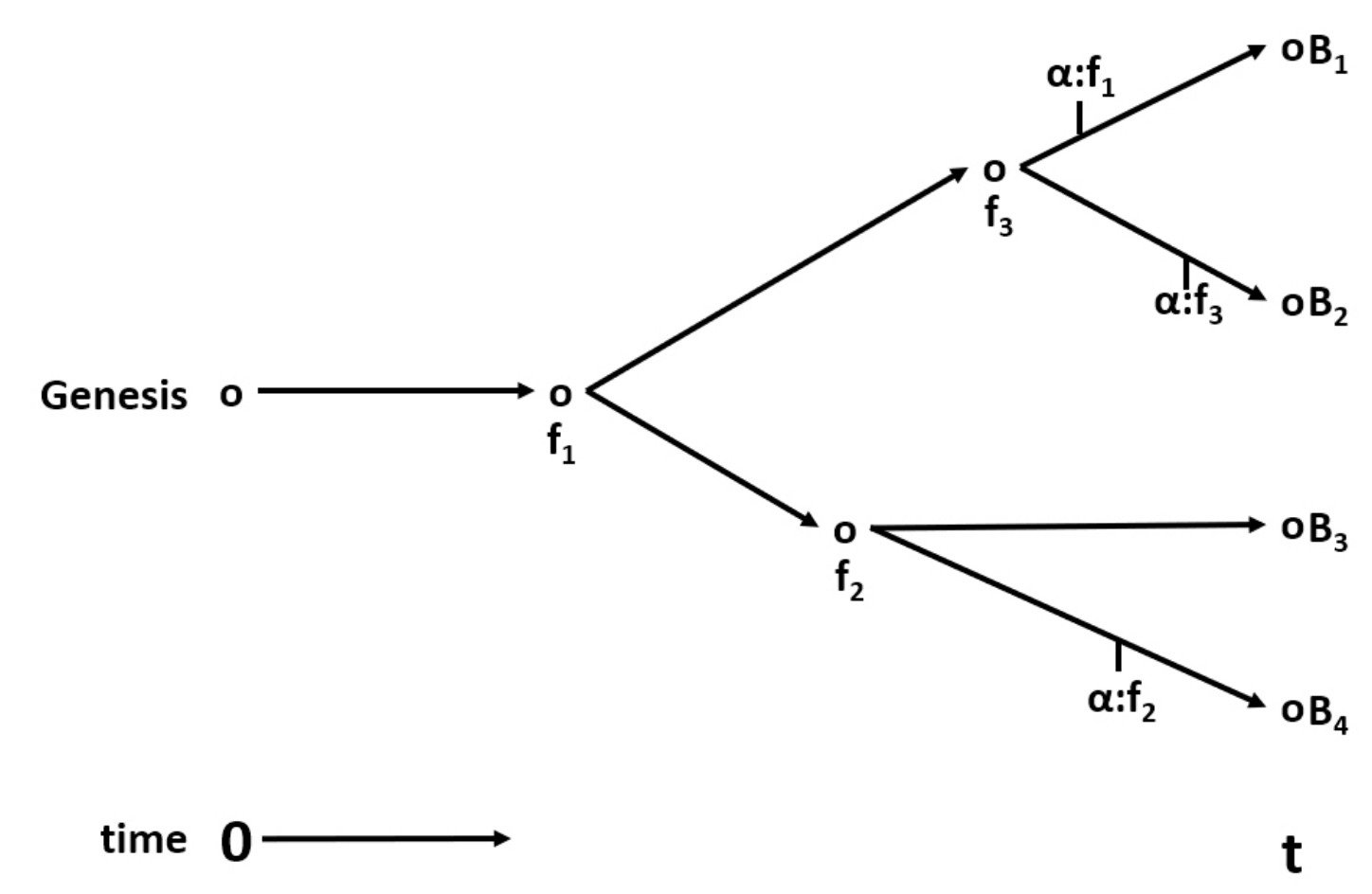}
\end{center}
\caption{A Blockchain Network}
\label{fig: A Blockchain Network}
\end{figure}

\subsection{Generalized penalty assignment rule and penalty function}

We restate the penalty assignment rule and penalty function for the general case.\vskip5pt

\textbf{Generalized Penalty Assignment Rule}

(i) Apply the $ADESS$ Penalty Assignment Rule to each fork-block $f_{n}$ (i.e. a chain that is not first to broadcast $\alpha$ post-fork blocks is assigned a penalty relative to the chain that is first to broadcast $\alpha$ post-fork blocks ) with the exception that\vskip5pt

(ii) If the chain that is first to broadcast $\alpha$ blocks after a fork-block $f_{n}$ is subject to an active penalty at the time of the broadcast, then no penalty assignment is made at fork-block $f_{n}$.\qed
\vskip5pt

\textbf{Generalized Penalty Function}

(i) Apply the $ADESS $ Penalty Function to each penalized chain at fork-block $f_{n}$.\vskip5pt

(ii) At the time a chain $B_{m}$ has overcome every penalty assigned to it, it has no active penalties and the protocol for $B_{m}$ reverts to the Nakamoto criteria of comparing cumulative mining puzzle difficulty with other chains that do not face active penalties. When it has overcome its last penalty, the score of $B_{m}$ is re-set to equal the cumulative mining puzzle difficulty of the baseline chain for the last penalty, plus a small additional amount $\epsilon$. For example, if $<B_{m}:f_{n}:B_{m'}:t>$ is the last penalty to become inactive, $B_{m}$ is assigned an adjusted cumulative mining puzzle difficulty = ($B_{m'}$ cumulative puzzle difficulty at $t$) + $\epsilon$. If $B_{m}$ has more than one active penalties and all are overcome at the same time, the baseline chain with the highest score is used for the re-set.\qed

\subsection{The canonical chain}

At any time the set of chains can be partitioned into two groups. One group are chains that have been assigned at least one penalty that has not been overcome. These chains are not eligible to be canonical. Among chains in the other group, those ranked by (possibly adjusted) cumulative mining puzzle difficulty, one will be canonical - provided there is at least one chain in this group. Proposition 2 establishes that there is at least one chain that is eligible to be canonical. 

\begin{proposition}
\label{prop: un-penalized chain}
Under the generalized $ADESS$ Protocol, there is at least one chain to which no penalty has ever been applied at any time.
\end{proposition}

\begin{proof}
We prove the proposition by construction. Start at the Genesis block and proceed to the first fork-block $f_{1}$. There will be at least one post-fork chain that is not penalized. Choose one of the non-penalized chains and proceed to the next fork-block. There will be at least one post-fork chain that is not penalized. Choose one of the non-penalized chains and proceed to the next fork-block, and so forth until a head $B_{m}$ is reached. The chain $B_{m}$ is not penalized at any fork.
\end{proof}

The example in the proof is displayed in Figure \ref{fig: A Blockchain Network} when comparing chains $B_{1}$ and $B_{2}$. If $B_{1}$ is the un-penalized chain at $f_{1}$, then either $B_{1}$ or $B_{2}$ must be un-penalized.  

\begin{corollary**}
There is exactly one canonical chain under generalized $ADESS$.
\end{corollary**}

\begin{proof}
Proposition 2 states that there is at least one chain to which no penalty has ever been applied at any time. Such a chain is eligible to be canonical. Suppose there is more than one chain without active penalties at a point in time. These chains are compared on the basis of cumulative mining puzzle difficulty. Under Nakamoto the chain with the most work is canonical.
\end{proof}

Finally, $ADESS$ applies to a circumstance where there is only one fork-block with actively mined descendant chains and none of those chains have forks. $ADESS$ is the application of generalized $ADESS$ in the case where there is one fork-block, two fork chains and neither chain has an active penalty from a prior fork. In that case, chain $\mathcal{A}$ is penalized relative to chain $\mathcal{IC}$.

\section{Relaxing the Restriction on the Growth Rate of Chain $\mathcal{A}$}
\label{app:Non-Constant Growth Rate}

The model limits the attacker to choosing a constant growth rate $\gamma$ for chain $\mathcal{A}$. We now relax that restriction and allow the attacker to choose the growth rate of each block $n$ on chain $\mathcal{A}$ as the function $\gamma(n,\xi)$. Equation \ref{eq: marginal condition of an attack} becomes 

\begin{multline}\tag{10'}
M(\xi) = \underset{\text{additional block reward}}{\underbrace{p_{B}}} - c\underset{\text{cost increase on current blocks}}{\underbrace{\sum_{n = 0}^{\lceil N(1 + \xi) -1\rceil}d/d\xi(\delta^{n/(1+\xi)}(1+\gamma(\xi,n))^{n})}} \\- \underset{\text{cost of additional block}}{\underbrace{c}}
\end{multline}

We do not evaluate all possible functional forms of $\gamma(n,\xi)$. We show that Theorem 1 and Corollary 1 continue to hold if the growth rate is a affine function of $\xi$. Let $\gamma(n,\xi) =  \rho + \xi f(n)$ for some scalar $\rho > 0$ and function $f(n) > 0$. The derivative for the discounted cost of an attack at $n$th block on chain $\mathcal{A}$ is

\begin{multline}{\nonumber}
d/d\xi(\delta^{n/(1+\gamma(\xi,n))}(1+ \gamma(\xi,n))^{n}) =\\ 
nf(n)(\xi f(n)+\rho)^{n-1}\delta^{n/(\xi f(n) + \rho +1)} \\
- nf(n)log(\delta)[(\xi f(n) + \rho)^{n} + 1]\delta^{n/(\xi f(n) + \rho + 1)}[\xi f(n) + \rho +1]^{-2}
\end{multline}

The expression is positive, noting that $log(\delta) < 0$, since $\delta \in (0,1)$. Therefore Equation 10' is negative. Noting that $n$, $\delta$ and $\xi$ are strictly positive, the expression is bounded away from zero, which implies that there is no upper bound to the expression. It follows that there is a value $\underline{\xi}$ for which $\pi(\underline{\xi}) < 0$, which proves Theorem 1 and Corollary 1.

\end{appendix}


\begin{thebibliography}{20}

\bibitem{Budish (2018)}
Budish, Eric (2018).
\textit{The Economic Limits of Bitcoin and the Blockchain}.
National Bureau of Economic Research Working Paper 24717, June 2018.
\url{https://www.nber.org/papers/w24717}

\bibitem{Buterin ESS}
Buterin, Vitalik (2014).
\textit{Proof of Stake: How I Learned to Love Weak Subjectivity}.
Ethereum Foundation Blog, November 25th, 2014.
\url{https://blog.ethereum.org/2014/11/25/proof-stake-learned-love-weak-subjectivity/}


\bibitem{Eyal and Sirer (2014)}
Eyal, Ittaly and Emin Gun Sirer (2018).
\textit{Majority is not enough: Bitcoin Mining is Vulnerable}.
Christin, N., Safavi-Naini, R. (eds) Financial Cryptography and Data Security. FC 2014. Lecture Notes in Computer Science, vol 8437. Springer, Berlin, Heidelberg. 
\url{https://doi.org/10.1007/978-3-662-45472-5_28}


\bibitem{Garay Bitcoin Backbone (2017)}
Garay, J., Kiayias, A., Leonardos, N. (2017).
\textit{The Bitcoin Backbone Protocol with Chains of Variable Difficulty}.In: Katz, J., Shacham, H. (eds) Advances in Cryptology – CRYPTO 2017. CRYPTO 2017. Lecture Notes in Computer Science(), vol 10401. Springer, Cham. \url{https://doi.org/10.1007/978-3-319-63688-7_10}.

\bibitem{Gervais et.al.}
Gervais, Arthur, Vasileios Glykantzis, Ghassan O. Karame, Hubert Ritzdorf, Karl Wurst and Srdjan Capkun (2016).
\textit{On the Security and Performance of Proof of Work Blockchains}
CCS '16: Proceedings of the 2016 ACM SIGSAC Conference on Computer and Communications Security pp. 3-16.
Association for Computing Machinery, New York, NY. United States

\bibitem{Guo and Ren}
Guo, Dongning and Ling Ren (2022).
\textit{Bitcoin's Latency - Security Analysis Made Simple}
arxiv.2203.06357v3
\url{https://arxiv.org/abs/2203.06357}

\bibitem{Leshno and Strack}
Leshno, Jacob and Philipp Strack (2020).
\textit{Bitcoin: An Axiomatic Approach and an Impossibility Theorem}.
American Economic Review: Insights 2020, 2(3): 269-286.
American Economic Association


\bibitem{Lovejoy (2020)}
Lovejoy, James (2020).
\textit{Reorgs on Bitcoin Gold: Counterattacks in the wild}.
Medium \url{https://medium.com/mit-media-lab-digital-currency-initiative/reorgs-on-bitcoin-gold-counterattacks-in-the-wild-da7e2b797c21}

\bibitem{MIT DCI}
MIT Media Lab Digital Currency Initiative reorg tracker (2020).
\textit{51\% attacks - reorg tracker}.
\url{https://dci.mit.edu/51-attacks}

\bibitem{Moroz et.al. (2020)}
Moroz, Daniel; Daniel Aronoff, Neha Narula and David Parkes (2020).
\textit{Double Spend Counterattacks}.
arXiv:2002.10736[cs.CR]
\url{https://doi.org/10.48550/arXiv.2002.10736}

\bibitem{Bitcoin White Paper (2008)}
Nakamoto, Satoshi (2008).
\textit{Bitcoin; A Peer - to - Peer Electronic Cash System}.
\url{https://bitcoin.org/bitcoin.pdf}.Original Bitcoin code \url{https://satoshi.nakamotoinstitute.org/code/}, 2008\\
Original Bitcoin code \url{https://satoshi.nakamotoinstitute.org/code/}

\bibitem{Schelling (1960)}
Schelling, Thomas C. (1960).
\textit{The Strategy of Conflict}.
Cambridge, Harvard University Press.

\bibitem{Singer (2020)}
Singer, Andrew (2020).
\textit{Fight fire with fire: MIT scholar suggests ETC counters 51\% attacks}. Cointelegraph Set. 15, 2020. \url{https://cointelegraph.com/news/fight-fire-with-fire-mit-scholar-suggests-etc-counters-51-attacks} 

\bibitem{Eth Yellow Paper (2021)}
Wood, Gavin (2021).
\textit{Ethereum: A secure Decentralised Generalised Transaction Ledger - Berlin Version d77a387 - 2022-4-26}.
\url{https://ethereum.github.io/yellowpaper/paper.pdf}

\end{thebibliography}
\end{document}